\documentclass[12pt, final, onecolumn ]{IEEEtran}

\usepackage{graphicx}
\usepackage{color}

\usepackage{mathtools}
\usepackage{amsfonts}
\usepackage{amsmath}
\usepackage{setspace}
\usepackage{enumerate}
\usepackage{verbatim}
\usepackage{amssymb}
\usepackage{amsbsy}

\usepackage{amsthm,thmtools}

\usepackage{setspace}
\usepackage{url}
\usepackage[dvips]{epsfig}

\newcommand{\myref}[1]{(\ref{#1})}

\declaretheorem[name={Example},qed={\lower-0.3ex\hbox{$\square$}} ] {Example}

\makeatletter
\renewenvironment{proof}[1][\proofname]{%
   \par\pushQED{\qed}\normalfont%
   \topsep6\p@\@plus6\p@\relax
   \trivlist\item[\hskip\labelsep\bfseries#1\@addpunct{.}]%
   \ignorespaces
}{%
   \popQED\endtrivlist\@endpefalse
} \makeatother

 \declaretheorem[name={Theorem}  ] {Theorem}

\declaretheorem[name={Problem}  ] {Problem} 
\declaretheorem[name={Corollary}  ] {Corollary}

\declaretheorem[name={Remark}  ] {Remark} \declaretheorem[name={Proposition}  ] {Proposition}

\newcommand {\R}{\mathbb R}

\newcommand{\be}{\begin{equation}}
\newcommand{\ee}{\end{equation}}
\newcommand{\Int}{\operatorname{{\mathrm int}}}

\newcommand{\LMD}{\lambda_0,\dots,\lambda_n}

\newcommand{\LMDR}{\lambda_1,\dots,\lambda_n}

\begin{document}
\title{Optimal Translation Along a Circular mRNA
}
\author{Yoram Zarai, Alexander Ovseevich and Michael Margaliot
 \IEEEcompsocitemizethanks{
\IEEEcompsocthanksitem Y.
 Zarai is with the School of Electrical Engineering, Tel-Aviv University, Tel-Aviv 69978, Israel. 
E-mail: yoramzar@mail.tau.ac.il
\IEEEcompsocthanksitem
A. Ovseevich is with the Institute for Problems in Mechanics, Russian Academy of Sciences, pr. Vernadskogo 101, 119526 Moscow, Russia.
E-mail: ovseev@gmail.com
\IEEEcompsocthanksitem
M. Margaliot (Corresponding Author) is with the School of Electrical Engineering and the Sagol School of Neuroscience, Tel-Aviv
University, Tel-Aviv 69978, Israel.  
E-mail: michaelm@eng.tau.ac.il  
}}

\maketitle

\begin{abstract}
 The ribosome flow model on a ring~(RFMR) is a deterministic   model for translation of a circularized~mRNA. 
We derive  a new \emph{spectral representation}
 for the optimal steady-state production rate and the corresponding optimal steady-state ribosomal density 
in the~RFMR.  This representation has several important advantages.
First,  it provides a simple and
numerically stable algorithm for determining  the optimal values even in very long rings.
Second,  it  enables   efficient
computation of the sensitivity of the optimal   production rate to small changes in the transition rates along the~mRNA. 
Third, it implies that the optimal steady-state production rate is a strictly concave function of the transition rates. Thus maximizing the optimal steady-state production rate with respect to the rates, under an affine constraint on the rates becomes a convex optimization problem that admits a unique solution, which can be determined numerically using highly efficient algorithms. This optimization problem 
 is   important, for example,  when re-engineering heterologous genes in a host organism. 
We describe  the implications of our results to this and other aspects of translation. 
 \end{abstract}


\section{Introduction}
Gene expression is the process by which the information encoded in a
 gene is used to  synthesize a functional gene product. Two main stages of this process are
  \emph{transcription} in which
 the information   in the DNA of a specific gene is copied into a  messenger RNA~(mRNA) molecule 
and  \emph{translation}. The latter  includes three phases: 
 (1) initiation: complex macro-molecules called
ribosomes bind   to the mRNA;(2)~elongation: the ribosomes    unidirectionally decode each codon   into
the corresponding amino-acid that is delivered to the awaiting ribosome by transfer RNA (tRNA) molecules; 
and (3)~termination: the ribosome detaches from the mRNA, the amino-acid sequence is released, folded and becomes a
functional protein~\cite{Alberts2002}. The output rate of ribosomes from the mRNA, which is also the rate in which
proteins are generated, is called the protein translation rate or production rate.

Translation occures in all living organisms, and under almost all conditions, to generate the macromolecular machinery
for life. Developing a deeper understanding of translation has important implications in numerous scientific disciplines
including medicine, evolutionary biology, biotechnology, and synthetic biology. Computational  models of translation are
essential in order to better understand this complex, dynamical and tightly-regulated
 process. Such models can also aid in integrating and analyzing the rapidly increasing
 experimental findings related to translation~(see, e.g., \cite{Dana2011,TullerGB2011,Tuller2007,Chu2012,Shah2013,Deneke2013,Racle2013,Zur2016}).

Computational models of translation describe the dynamical flow of ribosomes along the mRNA molecule, and include
parameters that encode the various factors affecting the codon decoding rates and the binding of ribosomes.  Some of
these  models provide a framework for both rigorous
 analysis and Monte Carlo simulations, thus promoting a better understanding of the way the parameters, and other factors, affect the dynamical and steady-state behavior of translation.  Several  computational
 models have been suggested based on different paradigms  for example  kinetics-based
ordinary differential equations~(see, e.g.~\cite{Na2010_BMC}), Petri  nets~\cite{Brackley2012128}, and
 probabilistic Boolean networks~\cite{Zhao2014}. For more details, see the survey papers~\cite{Haar2012,Zur2016}.

A standard mathematical model for ribosome flow  is the
  \emph{totally asymmetric simple exclusion process}~(TASEP) \cite{Shaw2003,TASEP_tutorial_2011}.
 In this model,  particles   hop unidirectionally
 along an ordered lattice of~$L$ sites. Each site can be either free or occupied by a particle, and a particle can only
   hop to a free site. This  \emph{simple exclusion principle}
   models particles that have ``volume'' and thus cannot overtake one  other.
   The hops are stochastic, and the rate of hoping from site~$i$ to site~$i+1$  is denoted by~$\gamma_i$.
   TASEP has two standard configurations. In TASEP with \emph{open boundary conditions}  the two sides of the lattice are connected to two particle reservoirs, and a particle can hop to [from] the first [last] site of the lattice at a rate~$\alpha$   [$\beta$].
   The average flow through the lattice converges to a steady-state value that depends on
    the parameters~$\alpha, \gamma_1,\dots,\gamma_{L-1},\beta$.
    Analysis of TASEP with open boundary conditions is  non trivial, and closed-form results have been obtained mainly for the homogeneous~TASEP~(HTASEP), i.e.
    for the case where all the~$\gamma_i$s are assumed to be equal.
  
    In TASEP with \emph{periodic boundary conditions} the chain is closed, and a particle that hops from the last site returns to the first one. Thus, here the lattice is  a ring, and the total
    number of particles along  the ring is conserved.

  TASEP has become a
 fundamental model in non-equilibrium statistical mechanics, and has been
 applied to model numerous natural and artificial  processes such as traffic flow, communication networks, and pedestrian dynamics~\cite{TASEP_book}.
In the context of translation, the lattice  models the mRNA molecule, the particles are ribosomes,
   and simple exclusion means that a ribosome  cannot overtake a ribosome in front of it.

The \emph{ribosome flow model}~(RFM)~\cite{reuveni} is a continuous-time  deterministic 
   model for the unidirectional  flow of ``material" along an open chain of $n$ consecutive compartments (or sites).
The~RFM can be derived via  a dynamic
 mean-field approximation of~TASEP with open boundary conditions~\cite[section 4.9.7]{TASEP_book}~\cite[p. R345]{solvers_guide}.
 In a RFM with $n$ sites, the state variable $x_i(t)\in[0,1]$, $i=1,\dots,n$, describes the normalized amount of ``material" (or density) at site~$i$ at time $t$, where $x_i(t)=1$ [$x_i(t)=0$] indicates that site~$i$ is completely full [completely empty] at time~$t$. In the~RFM, the two sides of the chain are connected to two particle reservoirs.
A   parameter~$\lambda_i>0$, $i=0,\dots,n$, controls the transition rate from site~$i$ to site~$i+1$,
 where~$\lambda_0$ [$\lambda_n$] is
the initiation [exit] rate. 

In the  \emph{ribosome flow model on a ring} (RFMR)~\cite{rfmr}
 the particles exiting the last site reenter the first site. This
 is a dynamic mean-field approximation of~TASEP with periodic boundary conditions.
The~RFMR admits a first integral, i.e. a quantity that is preserved along the dynamics, as the total ribosomal
 density  is conserved.
 Both the~RFM and~RFMR are cooperative dynamical systems~\cite{hlsmith}, but
their dynamical properties  are quite different~\cite{rfmr}.

Through simultaneous interactions with the cap-binding protein eIF4E and the poly(A)-binding protein~PABP, the
eukaryotic initiation factor~eIF4G  is able to bridge the two ends of the mRNA~\cite{Wells1988,presii2003}.  This
suggests that
 a  large fraction of the ribosomes that complete translating the mRNA  re-initiate.
The~RFMR is a good approximation of the translation dynamics in these circularized mRNAs.
In addition, circular RNA forms (which includes covalent RNA interactions) appear in all domains of
life~\cite{Danan2012,Cocquerelle1993,Cell1993,Burd2010,Hensgens1983,Abe2015,Granados-Riveron2016,AbouHaidar2014}, and it
was recently suggested that circular RNAs can be translated in
eukaryotes~\cite{Abe2015,Granados-Riveron2016,AbouHaidar2014}.

It was shown in~\cite{rfmr} that the~RFMR admits a unique steady-state  that depends on the total initial density and the
transition rates, but not on the distribution of the total initial density among the sites. All trajectories emanating
from initial conditions with the same total density converge to the unique steady-state. Ref.~\cite{RFM_r_max_density}
considered
 the ribosomal density along a circular mRNA that \emph{maximizes}
 the steady-state production rate using the~RFMR. It was shown that given any arbitrary set of positive transition
rates, there   exists a \emph{unique} optimal density (the same is true  for~TASEP with periodic boundary condition~\cite{Marshall2014}). However, this unique optimum was not given
explicitly, other than under certain special symmetry conditions on the rates, where the optimal density is one half of
the maximal possible density.

The ribosomal density along the mRNA molecule plays a critical role  in regulating gene expression, and specifically in determining protein production rates~\cite{picard2013significance,Benet2017}.
 For example, it was suggested in~\cite{Benet2017} that the cell tightly regulates ribosomal densities in order to maintain protein concentrations at different growth temperatures. At higher temperatures, the ribosomal density along the mRNA ``improves" in order to increase protein production rates (as protein stability decreases with temperature). 

The ribosomal density also affects different fundamental intracellular phenomena. Traffic james, abortions, and collisions may form if the ribosomal density is very high~\cite{Subramaniam2014}. It may also contribute to co-translational misfolding of proteins, which then requires additional  resources in order to degrade the  degenerated proteins~\cite{Drummond2008,Kurland1992,Zhang2009}.  On the other hand, a very low ribosomal density may lead to high degradation rate of mRNA molecules~\cite{Kimchi-Sarfaty2013,Edri2014,Tuller2015,Proshkin2010}. Thus, analyzing the ribosomal density that maximizes the production rate is critical in understanding how cells evolved to adapt and thrive in  a changing environment.

Here we derive  a new \emph{spectral representation}~(SR)
 for  the optimal steady-state production rate and the
corresponding steady-state ribosomal density in the~RFMR. This~SR  has several important
advantages. First, it provides a simple and
numerically stable way  to compute the optimal values even in very long rings. Second, 
 it  enables   efficient
computation of the sensitivity of the optimal steady-state production rate   to small changes in the transition rates.
This sensitivity analysis may find important applications in synthetic biology  where a crucial problem is
to determine the codons that are the most ``important'' in terms of their effect on the production rate.
Third, the~SR   implies that the optimal steady-state production rate is a strictly concave function of the~RFMR rates. Thus, the problem of maximizing the optimal steady-state production rate \emph{with respect to the rates} becomes a convex optimization problem that admits a unique solution, which can be determined numerically using highly efficient algorithms.

The remainder of this paper  is organized as follows. The next two sub-sections briefly review  the~RFM and the~RFMR.
Section~\ref{sec:main} describes our main results and  their biological implications. Section~\ref{sec:disscus} concludes and 
suggests  several directions for further research. To increase the readability of this paper, the proofs of all the
results are placed in the Appendix.

\subsection{Ribosome Flow Model (RFM)}
In a RFM with $n$ sites, the state variable $x_i(t)\in[0,1]$, $i=1,\dots,n$, denotes the density at site~$i$ at time~$t$,
where~$x_i(t)=1$ [$x_i(t)=0$] means that site~$i$ is completely full [empty] at time~$t$.
The~$n+1$ parameters~$\lambda_i>0$, $i=0,\dots,n$, control  the transition rate from site $i$ to site $i+1$. The RFM is a set of~$n$
first-order nonlinear ordinary differential equations describing the change in the amount of ``material'' in each site:
\be\label{eq:rfm_all} \dot{x}_i(t)=\lambda_{i-1}x_{i-1}(t)(1-x_i(t))-\lambda_i x_i(t) (1-x_{i+1}(t)),\quad i=1,\dots,n,
\ee 
where~$x_0(t):=1$, $x_{n+1}(t):=0$, and $\dot{x}_i$ is the change in the amount of material at site $i$ at time $t$, i.e.
$\dot{x}_i(t):=\frac{d}{dt}x_i(t)$, $i=1,\dots,n$.
Eq.~\eqref{eq:rfm_all} can be explained as a kind of a master equation: the change in density in site~$i$
is the flow from site~$i-1$ to site~$i$ minus the flow from site~$i$ to site~$i+1$.
 The first flow, that is, the input rate to site $i$ is~$\lambda_{i-1} x_{i-1}(t)(1 - x_{i}(t)
)$. This rate is proportional to $x_{i-1}(t)$, i.e. it increases
    with the density at site~$i-1$, and to $(1-x_{i}(t))$, i.e. it decreases as site~$i$ becomes fuller.
In particular, when site $i$ is completely full, i.e. when~$x_{i}(t)=1$, there is no flow into this site.
 This is reminiscent of the  simple exclusion principle: the flow of
particles into a site decreases as that site becomes fuller. Note that the maximal possible
  flow  from site~$i-1$ to site~$i$  is~$\lambda_{i-1}$. Similarly, the output rate from site~$i$, which is also the input rate to site~$i+1$, is given by $\lambda_{i} x_{i}(t)(1 - x_{i+1}(t) )$.
The output rate from the chain is  $R(t):=\lambda_n x_n(t)$, that is, the flow out of the last site. 

In the context of mRNA translation, the $n$-sites chain is   the mRNA, $x_i(t)$ describes the ribosomal density at site~$i$ at time~$t$, and~$R(t)$ describes the rate at which  ribosomes leave the mRNA, which is also the rate at which
the  proteins are generated. Thus,~$R(t)$ is   the protein translation rate or production rate at time~$t$.

Since every state-variable models the density of ribosomes in a site,   normalized such that a value zero [one]
 corresponds to a completely empty [full] site, the state space of the~RFM is the~$n$-dimensional unit cube $C^n:=[0,1]^n$.
Let~$x(t,a)$ denote the solution of the RFM at time~$t$ for the initial condition~$x(0)=a$. It has been shown
in~\cite{RFM_stability}  (see also~\cite{RFM_entrain}) that  for every~$a \in C^n$ this solution  remains in~$C^n$ for
all~$t\geq 0$, and that the~RFM admits a globally asymptotically stable steady-state $e\in\Int(C^n)$, i.e.
$\lim_{t\to \infty} x(t,a)=e$ for all~$a \in C^n$. The value~$e$ depends on the rates $\LMD$, but not on the initial
condition~$x(0)=a$.
  This means that if we simulate the~RFM starting from any initial density
 of ribosomes on the mRNA the dynamics will
 always converge to the same steady-state (i.e., to the same final ribosome density along
 the mRNA).
In particular,  the production rate~$R(t)=\lambda_n x_n(t)$ always converges to the steady-state value: \be
\label{eq:defr} R:=\lambda_n  {e}_n. \ee

A spectral representation of this steady-state value has been derived in~\cite{rfm_max}. Given a~RFM with dimension~$n$
and rates~$\lambda_0,\dots,\lambda_n$, define a $(n+2)\times(n+2)$ Jacobi matrix
\be\label{eq:bmatb}
                B(\LMD) := \begin{bmatrix}
 0 &  \lambda_0^{-1/2}   & 0 &0 & \dots &0&0 \\
\lambda_0^{-1/2} & 0  & \lambda_1^{-1/2}   & 0  & \dots &0&0 \\
 0& \lambda_1^{-1/2} & 0 &  \lambda_2^{-1/2}    & \dots &0&0 \\
 & &&\vdots \\
 0& 0 & 0 & \dots &\lambda_{n-1}^{-1/2}  & 0& \lambda_{n }^{-1/2}     \\
 0& 0 & 0 & \dots &0 & \lambda_{n }^{-1/2}  & 0
 \end{bmatrix}.
\ee Note that~$B$ is componentwise non-negative and irreducible, so it admits a Perron root~$\mu>0$. It has been shown
in~\cite{rfm_max} that~$\mu=R^{-1/2}$. This provides a way to compute the steady-state~$R$ in the RFM without simulating the
dynamical equations of the~RFM.

For more on the analysis of the~RFM  using tools from systems and control theory and the biological implications of this
analysis, see~\cite{zarai_infi,rfm_max,RFM_sense,RFM_feedback,rfm_control,rfm_opt_down}. Recently, a network of~RFMs,
interconnected via a pool of ``free'' ribosomes,
 has been used to model and analyze competition for ribosomes in the cell~\cite{RFM_model_compete_J}.

\begin{figure*}[t]
 \begin{center}
  \includegraphics[scale=0.8]{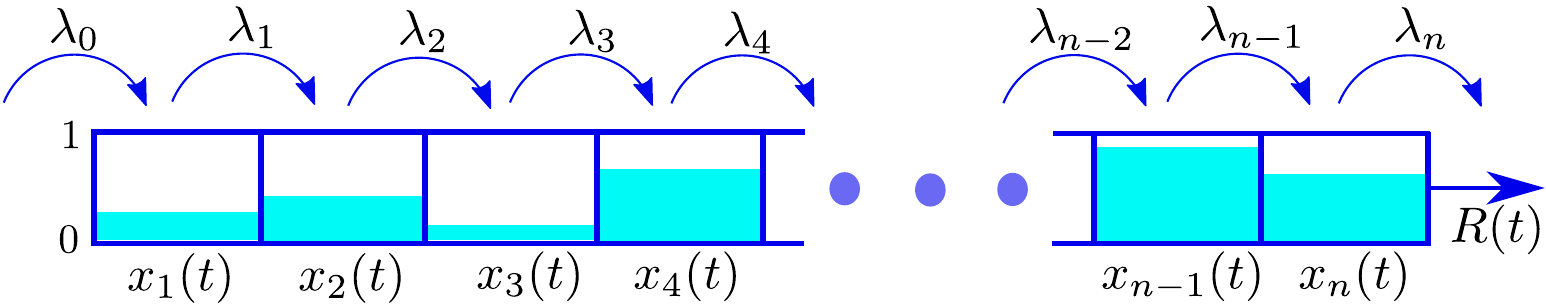}
\caption{The RFM models unidirectional flow along a chain of~$n$ sites. The state variable~$x_i(t)\in[0,1]$ represents
 the density at site $i$ at time $t$. The parameter $\lambda_i>0$ controls the transition rate from  site~$i$ to site~$i+1$, with~$\lambda_0>0$ [$\lambda_n>0$] controlling    the initiation [exit] rate. The output rate at time $t$ is~$R(t) :=\lambda_n x_n(t)$. }\label{fig:rfm}
\end{center}
\end{figure*}

\subsection{Ribosome Flow Model on a Ring (RFMR)}
If we consider the RFM under  the additional assumption that all the ribosomes leaving site~$n$ circulate back to
site~$1$ then we obtain the~RFMR (see Fig.~\ref{fig:rfmr}). Just like the RFM, the RFMR is described by $n$ nonlinear,
first-order ordinary differential equations:
\begin{align}\label{eq:rfmr}
                    \dot{x}_1&=\lambda_n x_n (1-x_1) -\lambda_1 x_1(1-x_2), \nonumber \\
                    \dot{x}_2&=\lambda_{1} x_{1} (1-x_{2}) -\lambda_{2} x_{2} (1-x_3) , \nonumber \\
                             &\vdots \nonumber \\
                    \dot{x}_n&=\lambda_{n-1}x_{n-1} (1-x_n) -\lambda_n x_n (1-x_1) .
\end{align}
The    difference here with respect  to the~RFM is in the equations describing the change of material in sites~$1$
and~$n$. Specifically, the flow out of site~$n$ is the flow into site~$1$. 
This model assumes  perfect recycling (be it covalent or non-covalent),  and    provides a
 good approximation   when a large fraction of the ribosomes are recycled.
 Note that the RFMR can  also be written succinctly  as~\eqref{eq:rfm_all}, but now  with every index   interpreted modulo~$n$. In particular, $\lambda_0$
[$x_0$] is replaced by~$\lambda_n$ [$x_n$].

\begin{figure*}[t]
 \begin{center}
  \includegraphics[scale=0.8]{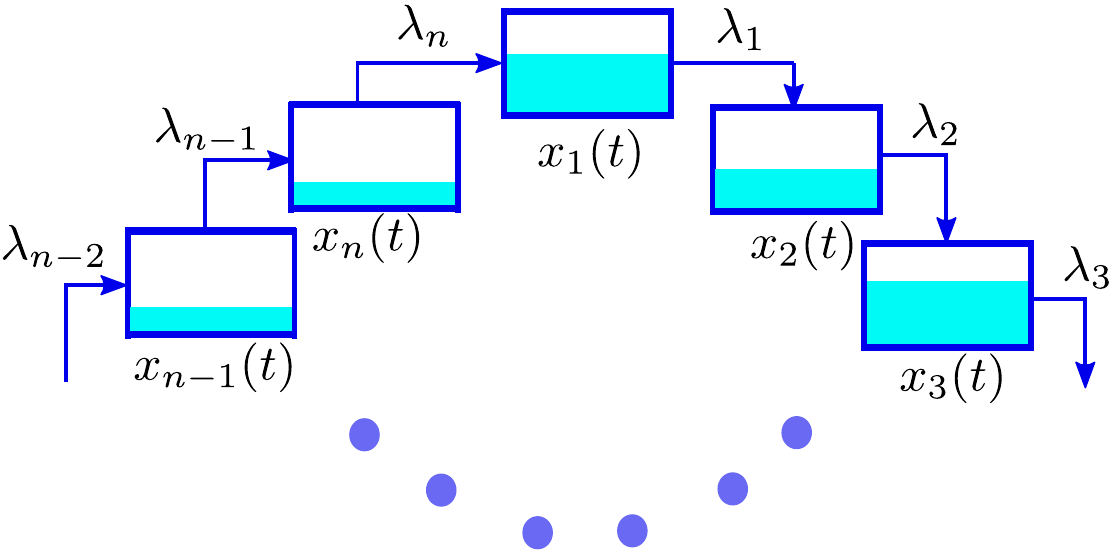}
\caption{The RFMR models unidirectional flow along a circular chain of~$n$ sites. The parameter $\lambda_i>0$ controls
the transition rate from site $i$ to site $i+1$.   }\label{fig:rfmr}
\end{center}
\end{figure*}

\begin{Remark}\label{remk:cir_rfmr}
It is clear from the cyclic topology of  the~RFMR that if we cyclically shift all the rates~$k$  times for some integer~$k>1$
 then the model
does not change.
\end{Remark}

In the RFMR the total density of ribosomes along the  ring   at time~$t$ is given by
\[
H(x(t)):=x_1(t)+\dots+x_n(t),
\]
i.e. the sum of the density at each site. Let~$s$ denote  the total density of ribosomes along the  chain at time~$t=0$,
i.e. $s:=H(x(0))$. Since ribosomes that exit site $n$ circulate back to site $1$, the total density  is preserved for all
time, that is,~$H(x(t))\equiv s$ for all~$t\geq 0$.
  The dynamics of the RFMR thus redistributes the particles between the sites,
but without changing  the total ribosome density.
 In the context of translation, this means that
the total number of ribosomes on the  (circular) mRNA is conserved. We say  that~$H(x(t))$ is a \emph{first integral} of
the~RFMR.

For~$s\in[0,n]$, the \emph{$s$ level set} of~$H$ is
\[
L_s:=\{  y \in C^n: y_1+\cdots+y_n= s   \}.
\]
This is the  set of all possible ribosome density configurations such that the total density is equal to~$s$. For
example, the vectors of densities~$\begin{bmatrix} s&0  &0 &\dots&0 \end{bmatrix}'$ and~$\begin{bmatrix} s/2&s/2  &0
&\dots&0 \end{bmatrix}'$ both belong to~$L_s$.

Ref.~\cite{rfmr} has shown that the~RFMR is a strongly cooperative
 dynamical system, and that this
implies that every level set~$L_s$  contains a unique steady-state~$e=e(s,\lambda_1,\dots,\lambda_n)\in\Int(C^n)$,
and that any trajectory of the~RFMR  emanating  from any~$x(0)\in L_s$ converges to this steady-state point. In
particular, the production rate converges to a steady-state value $R=R(s,\LMDR)$.

 Pick~$s\in[0,n]$ and~$a\in L_s$.
Consider the~RFMR with~$x(0)=a$. Let
\[
\rho:=s/n
\]
 denote the average ribosomal density in the RFMR.
At steady-state, i.e. for $x=e$ the left-hand side of all the equations in~\eqref{eq:rfmr} is zero, so
\begin{align}\label{eq:rfmr_e}
                     \lambda_n e_n (1-e_1) &=\lambda_1 e_1(1-e_2), \nonumber \\
                                           &=\lambda_{2} e_{2} (1-e_3) , \nonumber \\
                             &\vdots \nonumber \\
                     &=\lambda_{n-1}e_{n-1} (1-e_n), \nonumber \\
                     &=R,
\end{align}
and, since the total density is conserved, 
\[
e_1+\dots+e_n = s.
\]

Note that it follows from~\eqref{eq:rfmr_e} that for any $c>0$
\be\label{eq:hom_deg1}
R(s,c\lambda_1,\dots,c\lambda_n)=cR(s,\LMDR),
\ee
i.e. if we multiply all the rates by a factor $c>0$ then the steady-state production rate will also increase by the same factor $c$.

Given a set of transition rates, an interesting question is what ribosomal density \emph{maximizes}
 the steady-state production
rate in the~RFMR? Indeed, a   ribosomal density~$s=0$  means zero production rate (as there are no ribosomes on the ring), and so
does the  completely full density~$s=n$, as all the sites are completely full and the ribosomes cannot move forward. 
It was shown in~\cite{RFM_r_max_density} that for any arbitrary positive set of rates~$\LMDR$,
 there exists a \emph{unique}   density~$s^*=s^*(\LMDR)$ (and thus a unique average density~$\rho^*=s^*/n$) that maximizes the steady-state production rate. We denote the corresponding optimal
 steady-state production rate by $R^*=R(s^*(\LMDR),\LMDR)$, and the corresponding optimal
steady-state density by $e^*=e(s^*(\LMDR),\LMDR)$. This means that in order to maximize the steady-state production rate (with respect to the total density), the
mRNA must be initialized with a total density~$s^*$ (the  distribution of this total density  along the mRNA at time zero
is not important). Initializing with either more or less than $s^*$  (i.e with~$\sum_{i=1}^n x_i(0)>s^*$ or with~$\sum_{i=1}^n x_i(0)<s^*$ ) will  decrease the   steady-state production rate
with respect to the one obtained when the circular mRNA is initialized with total density~$\sum_{i=1}^n x_i(0)=s^*$.

The results in~\cite{RFM_r_max_density} also show that for the optimal value~$s^*$, the steady-state density satisfies:
\be\label{eq:swee(1-e)} \prod_{i=1}^n e^*_i=\prod_{i=1}^n (1-e^*_i). \ee This can be explained as follows. If the total
density~$s$ is too small then the product~$\prod_{i=1}^n e_i$ is also small and thus~$ \prod_{i=1}^n e_i< \prod_{i=1}^n
(1-e_i)$. This case is not optimal i.e. it does not maximizes~$R$, as there are not enough ribosomes on the ring.  If the
total density~$s$ is too
 large then a similar argument yields~$ \prod_{i=1}^n e_i> \prod_{i=1}^n (1-e_i)$.
 This case is also not optimal,
as there are too many ribosomes on the ring and this leads to ``traffic jams'' that reduce the production rate. The
optimal scenario  lies between these two cases and is characterized by~\eqref{eq:swee(1-e)}.

\begin{Example}\label{exa:s3}
  Fig.~\ref{fig:R_s_n3} depicts $R$ as a function of $s$ for a RFMR with dimension $n=3$ and rates $\lambda_1=0.7$, $\lambda_2=1.6$, and $\lambda_3=2.2$. It may be seen
    that there exists a unique value~$ s^*=1.4948$ (all numerical results in this paper are to four digit accuracy)
    that maximizes~$R$.
Simulating the~RFMR with this initial density (e.g., by
 setting~$x(0)=\begin{bmatrix}s^*/2   & 0   & s^*/2 \end{bmatrix}'$) yields
    \[
    e^*=\begin{bmatrix} 0.6878   & 0.3546   & 0.4524 \end{bmatrix}',
    \]
    and~$R^*=\lambda_3 e^*_3(1-e^*_1)= 0.3107$.
     Note that $s^*$ is   close (but not equal) to~$3/2$, that is,  one half of the maximal density.
    Note also that~$\prod_{i=1}^3 e^*_i=\prod_{i=1}^3 (1-e^*_i)=0.1103 $.
\end{Example}

\begin{figure}[t]
 \begin{center}
  \includegraphics[scale=0.5]{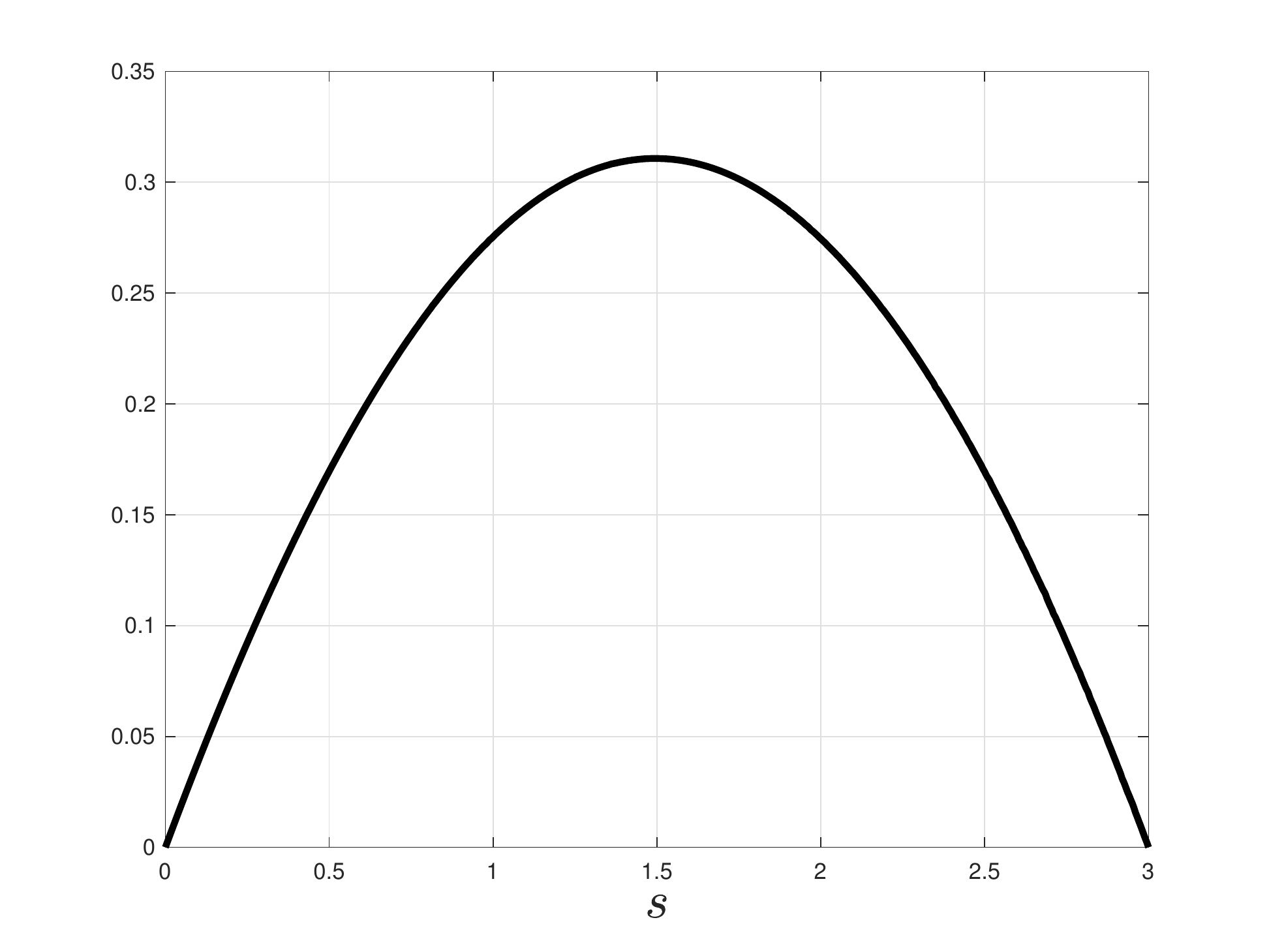}
\caption{Steady-state production rate~$R$ as a function of the total (conserved) ribosomal density~$s\in[0,3]$, for a RFMR with
dimension $n=3$ and transition  rates $\lambda_1=0.7$, $\lambda_2=1.6$, and $\lambda_3=2.2$.}\label{fig:R_s_n3}
\end{center}
\end{figure}

Here, we present for the first time a \emph{spectral representation}
 of the optimal steady-state production rate~$R^*$ and  the steady-state density~$e^*$ in the
 non-homogeneous~RFMR. We show that this representation
 has several advantages. First, it provides an efficient and numerically stable algorithm for evaluating $R^*$ and $e^*$ (and thus $s^*$)
even for very large rings.  This completely
 eliminates the need to   simulate the~RFMR dynamical equations for different values
of~$s$ in order to determine the optimal
 values.
 Furthermore, the spectral representation allows to analyze the sensitivity of~$R^*$ to small changes in the rates. This sensitivity analysis could be crucial for example in synthetic biology applications, where an important problem is to determine positions along the transcript that affect the production rate the most (as this is not necessarily the positions of the slowest codons)~\cite{gene_cloning_book}.
Finally, we show that the spectral representation implies that $R^*$ is a strictly concave function of the rates. This means that the problem of maximizing~$R^*$ \emph{with respect to the rates} is a convex optimization problem, thus it admits a unique solution that can be efficiently determined numerically using algorithms that scale well with $n$. 

It is important to note that in general the analysis results for the~RFMR  hold for any set of transition
rates. This is in contrast to the analysis results for TASEP. Rigorous analysis of TASEP seems to be tractable
only under the assumption that the internal hopping rates are all equal (i.e. the homogeneous case).
In the context of translation, this models the very special case where all elongation rates are assumed to be equal.

The next
 section derives a spectral representation for~$e^*$ and~$R^*$, and describes its  implications.

\section{Main  Results}\label{sec:main}

\subsection{Spectral Representation}

Consider a RFMR with dimension $n>2$ and rates $\LMDR>0$. Define an $n\times n$ matrix  
\be\label{eq:matA}
A(\LMDR) := \begin{bmatrix} 
 0 &  \lambda_1^{-1/2}   & 0 &0 & \dots &0&\lambda_n^{-1/2} \\
\lambda_1^{-1/2} & 0  & \lambda_2^{-1/2}   & 0  & \dots &0&0 \\
 0& \lambda_2^{-1/2} & 0 &  \lambda_3^{-1/2}    & \dots &0&0 \\
 & &&\vdots \\
 0& 0 & 0 & \dots &\lambda_{n-2}^{-1/2}  & 0& \lambda_{n-1 }^{-1/2}     \\
 \lambda_n^{-1/2}& 0 & 0 & \dots &0 & \lambda_{n-1 }^{-1/2}  & 0
 \end{bmatrix}.
\ee Note that this is a periodic Jacobi matrix (see, e.g.~\cite{pjm_1980}).

We use the notation~$\R^n_{++}:=\{v\in\R^n : v_i>0, i=1,\dots,n\}$, that is, the set of all~$n$-dimensional vectors
with  positive entries. Since~$A$ is symmetric, all its eigenvalues are real. Since~$A$ is (componentwise) non-negative  and irreducible,   it admits a unique maximal eigenvalue~$\sigma>0$ (called the Perron eigenvalue or Perron root),
 and a corresponding eigenvector~$\zeta\in\R^n_{++}$ (the Perron eigenvector)~\cite{matrx_ana}.

Our first result provides a  representation for the optimal steady-states in the~RFMR using the spectral properties of
the matrix~$A$. In what follows, all indexes are interpreted modulo~$n$. Recall that all the steady-state properties are invariant to any
arbitrary cyclic shifts of the rates (see Remark~\ref{remk:cir_rfmr}), and that the proofs of all the results
are placed in the Appendix.

\begin{Theorem}\label{thm:rfmr_lin}
Consider a RFMR with dimension $n$ and rates $\LMDR$. Let~$\sigma>0$ $[\zeta \in \R^n_{++} ]$  denote the Perron
eigenvalue [eigenvector]
 of~$A$ in~\eqref{eq:matA}.
 Then the optimal values in the~RFMR satisfy:
\begin{align}\label{eq:spect_rep}
R^*&=\sigma^{-2}, \nonumber \\
e_i^*&=\lambda_i^{-1/2}\sigma^{-1}\frac{\zeta_{i+1}}{\zeta_i}, \quad i=1,\dots,n, \nonumber \\
s^*&= \sigma^{-1}\sum_{i=1}^n  \lambda_i^{-1/2} \frac{\zeta_{i+1}}{\zeta_i}.
\end{align}
\end{Theorem}


\begin{Example}\label{example2}
Consider a RFMR with dimension $n=3$ and rates $\lambda_1=0.7$, $\lambda_2=1.6$, and $\lambda_3=2.2$. The corresponding
matrix $A$ is:
\[
A=\begin{bmatrix}
 0 & 1.1952 & 0.6742 \\
 1.1952 & 0 & 0.7906 \\
 0.6742 & 0.7906 & 0
\end{bmatrix}.
\]
The maximal eigenvalue of $A$ is $\sigma=1.7940$, and the corresponding eigenvector is $$\zeta=\begin{bmatrix} 0.6024 &
0.6219 & 0.5004 \end{bmatrix}'.$$ Now~\eqref{eq:spect_rep} yields $R^*=1.7940^{-2}=0.3107$, $e_1^*=0.6878$,
$e_2^*=0.3546$, $e_3^*=0.4524$, and~$s^*=1.4948$. This agrees of course with the results in Example~\ref{exa:s3}.
\end{Example}

Thm.~\ref{thm:rfmr_lin} thus provides a spectral
 representation of the optimal values~$R^*$, $e^*$, and~$s^*$.  This is important, since $e^*$ cannot be easily calculated
 based on the steady-state equations of the RFMR. Using simple and efficient algorithms to determine the eigenvalues and eigenvectors of a periodic Jacobi matrix, it is now possible to numerically compute $R^*$, $e^*_i$, $i=1,\dots,n$, and   $s^*$ even for very large rings and without any simulations of the dynamical equations of the RFMR.

Thm.~\ref{thm:rfmr_lin} has several more  interesting implications. 
Given a~RFMR with rates~$\lambda_1,\dots,\lambda_n$, define a vector~$\bar \lambda\in\R^n_{++}$ by $\bar
\lambda_1:=\lambda_2$, $\bar \lambda_2:=\lambda_3,\dots,\bar \lambda_n=\lambda_1$. In other words,~$\bar \lambda$ is a 1-step cyclic shift
of~$\lambda$. Let $P\in\R^{n\times n}$ be a matrix of zeros, except for the super-diagonal and the~$(n,1)$ entry that are
all equal to $1$. For example, for $n=4$, $
P=\begin{bmatrix} 0 & 1 & 0 & 0 \\
0 & 0  & 1 & 0 \\
0 & 0 & 0 & 1 \\
1 & 0 & 0 & 0
\end{bmatrix}.
$ Then~$P$ is a permutation matrix so that~$P'=P^{-1}$, and~$\bar  \lambda=P\lambda$. It is straightforward to show
that~$A(\bar \lambda)= PA(\lambda)P'$, so~$A(\lambda)$
 and~$A(\bar \lambda)$ have the same spectral properties. Thus, Thm.~\ref{thm:rfmr_lin}
leads to the same steady-state results for both the original RFMR and its cyclic shift and this agrees with
 Remark~\ref{remk:cir_rfmr}.

In some special cases, the Perron eigenvalue and eigenvector of~$A$ may be known explicitly
 and then one
can immediately determine the optimal steady-state in the corresponding~RFMR. The next example demonstrates this.
\begin{Example}\label{exa:homg}
Consider a RFMR with homogeneous transition rates, i.e. \be\label{eq:homog_rates}
\lambda_1=\cdots=\lambda_n:=\lambda_c, \ee where $\lambda_c$ denotes the common value of all the rates. Then it is
straightforward to verify that $A(\lambda_c,\cdots,\lambda_c)$  admits a Perron eigenvalue $\sigma=2 \lambda_c^{-1/2}$
and a corresponding   eigenvector $\zeta=\begin{bmatrix} 1 & 1 & \cdots & 1 \end{bmatrix}'$.  Thm.~\ref{thm:rfmr_lin}
implies that~$R^*=\lambda_c/4$ and $e_i^*=1/2$, $i=1,\dots,n$. This result has already been proven
in~\cite[Prop.~3]{RFM_r_max_density} using a different approach.
\end{Example}

\subsection{Steady-State  RFM as a Special Case of the Steady-State RFMR with Optimal Total Density}
Comparing the spectral representations for the~RFMR and the~RFM yields the following result. 
Consider a~RFMR with dimension~$n$, fixed rates~$\lambda_1,\dots,\lambda_{n-1}$, and~$\lambda_n\to \infty$. In this case,
the matrix~$A(\lambda)$ in~\eqref{eq:matA} converges to the matrix: \be\label{eq:matconv}
  \begin{bmatrix}
 0 &  \lambda_1^{-1/2}   & 0 &0 & \dots &0&0 \\
\lambda_1^{-1/2} & 0  & \lambda_2^{-1/2}   & 0  & \dots &0&0 \\
 0& \lambda_2^{-1/2} & 0 &  \lambda_3^{-1/2}    & \dots &0&0 \\
 & &&\vdots \\
 0& 0 & 0 & \dots &\lambda_{n-2}^{-1/2}  & 0& \lambda_{n-1 }^{-1/2}     \\
  0& 0 & 0 & \dots &0 & \lambda_{n-1 }^{-1/2}  & 0
 \end{bmatrix}.
\ee Comparing this with~\eqref{eq:bmatb} and using Thm.~\ref{thm:rfmr_lin}
 implies the following.
\begin{Corollary}
Let~$e^*=\begin{bmatrix}  e^*_1&\dots&e^*_n \end{bmatrix}'$ denote the   optimal steady-state of a~RFMR  with dimension~$n$ and rates~$\LMDR$, where $\lambda_n\to\infty$. 
Let~$\tilde e=\begin{bmatrix} \tilde e_1&\dots&\tilde e_{n-2} \end{bmatrix}'$ denote the steady-state
of a~RFM with dimension~$n-2$ and transition
 rates $\tilde \lambda_0=\lambda_1,\tilde \lambda_1=\lambda_2,\dots,\tilde \lambda_{n-2}=\lambda_{n-1}$. 
Then~$\tilde e=\begin{bmatrix}
e_2^*&e_3^*&\dots& e_{n-1}^* \end{bmatrix}'$.
\end{Corollary}

This implies that the steady-state of a RFM with arbitrary
 dimension~$m$ and arbitrary  rates~$\tilde \lambda_i>0$ can be derived from the steady-state of an RFMR with dimension $n:=m+2$, rates $\lambda_i=\tilde \lambda_{i-1}$, $i=1,\dots,n-1$, $\lambda_n\to\infty$, that is initialized with the optimal total density $s^*$. In this respect,  the~RFM is a kind of 
   ``open-boundaries''~RFMR that is  initialized with the optimal total density.

This connection between the two models can be explained  as follows. By~\eqref{eq:rfmr},  in an RFMR
with~$\lambda_n\to\infty$, the steady-state density at site~$n$ will be zero, and at site~$1$ it will be one. Indeed, the
transition rate from site~$n$ to site~$1$ is infinite, so site~$n$ will be completely emptied and site~$1$ completely
filled.
 This  ``disconnects'' the ring  at the link from site $n$ to site $1$. Furthermore,
  the completely  full site~$1$ serves as a ``source'' to site~$2$
    whereas the completely empty site~$n$ serves as a ``sink'' to site~$n-1$.
 The result is that sites~$2,\dots,n-1$, of the RFMR become   a~RFM  with dimension~$n-2$.
 The next example demonstrates this.

\begin{Example}
Consider a RFMR with dimension $n=5$, and rates
$\lambda_1=0.8,\lambda_2=0.6,\lambda_3=0.4,\lambda_4=0.7$, and $\lambda_5=0.5$. The optimal steady-state values
are:
\[
e^*=\begin{bmatrix} 0.4260 & 0.5831 & 0.5939 & 0.4019 & 0.4950 \end{bmatrix}', \;\; R^*=0.1421.
\]
For $\lambda_5=100$, the optimal steady-state values are:
\[
e^*=\begin{bmatrix} 0.9440 &0.7628 & 0.6087 &0.2643 &0.0320 \end{bmatrix}', \;\; R^*=0.1791,
\]
for $\lambda_5=10,000$, the optimal steady-state values are:
\[
e^*=\begin{bmatrix} 0.9942 & 0.7727 & 0.6100 & 0.2591 &0.0031 \end{bmatrix}', \;\; R^*=0.1808,
\]
and for $\lambda_5=1,000,000$, they are:
\be\label{eq:exp_l5large}
e^*=\begin{bmatrix} 0.9994 & 0.7737 & 0.6102 & 0.2586 & 0.0003 \end{bmatrix}', \;\; R^*=0.1810.
\ee
It may be observed that as $\lambda_5$ increases, the optimal steady-state density at site $5$ [site~$1$] decreases
[increases] to zero [one]. On the other hand, for a RFM with dimension $n=3$ and rates $\tilde \lambda_0=0.8, \tilde
\lambda_1=0.6, \tilde \lambda_2=0.4$, and $\tilde \lambda_3=0.7$, the steady-state values are:
 $
\tilde e=\begin{bmatrix} 0.7738 & 0.6102 & 0.2585 \end{bmatrix}'$, and $\tilde R=0.1810$ (compare   to~\eqref{eq:exp_l5large}).
\end{Example}


\subsection{Sensitivity Analysis}

We already know that given the transition rates~$\lambda_1,\dots,\lambda_n$, the RFMR admits a unique
density~$s^*(\lambda_1,\dots,\lambda_n)$ for which the steady-state production rate is maximized.
Maximizing the steady-state production rate is a standard goal in biotechnology, and since codons may be replaced by their
synonymous, an important question in the context of the~RFMR is: how will a change in the rates affect the maximal
production rate~$R^*$? Note that the effect here is compound, as changing the rates also changes the optimal density that
yields the maximal production rate.

In this section, we analyze \be\label{eq:si_sens} \phi_i(\LMDR):=\frac{\partial}{\partial \lambda_i}R^*(\LMDR), \quad
i=1,\dots,n, \ee i.e. the sensitivity of the optimal steady-state production rate $R^*$ with respect to~$\lambda_i$.

A relatively large value of $\phi_i$ indicates that a small change in $\lambda_i$ will have a strong impact on the optimal
steady-state production rate $R^*$. In other words, the sensitivities  indicate
 which rates are the most ``important'' in
terms of their effect on $R^*$. The results in Thm.~\ref{thm:rfmr_lin} allow
 to compute the sensitivities using the spectral properties of the matrix~$A$.
\begin{Proposition}\label{prop:sense}
The sensitivities satisfy: 
\be\label{eq:R_sens}
\phi_i=\frac{2  \zeta_i \zeta_{i+1}}{\sigma^{3} \lambda_i^{3/2} \zeta' \zeta}, \quad i=1,\dots,n. 
\ee
\end{Proposition}
Eq.~\eqref{eq:R_sens} provides an efficient and numerically stable method to calculate
 the sensitivities   for large-scale rings and arbitrary positive rates~$\lambda_i$s
using standard algorithms for computing the eigenvalues and eigenvectors of periodic Jacobi matrices.
 Note that~\eqref{eq:R_sens} implies that all the sensitivities are positive.

\begin{Example}
 Fig.~\ref{fig:sens_n98_2} depicts $\ln(\phi_i)$, computed using~\eqref{eq:R_sens},
 as a function of~$i$ for a~RFMR with dimension~$n=98$  and   rates~$\lambda_1=\lambda_{50}=0.3$
and~$\lambda_i=1$ for all other~$i$. Here the maximal sensitivity is~$\phi_1=\phi_{50}$, and the sensitivities
decrease as we move away from sites~$1$ and~$50$. This makes sense
 as the corresponding rates are the bottleneck rates in this example. 
\end{Example}

\begin{figure}[t]
 \begin{center}
  \includegraphics[scale=0.5]{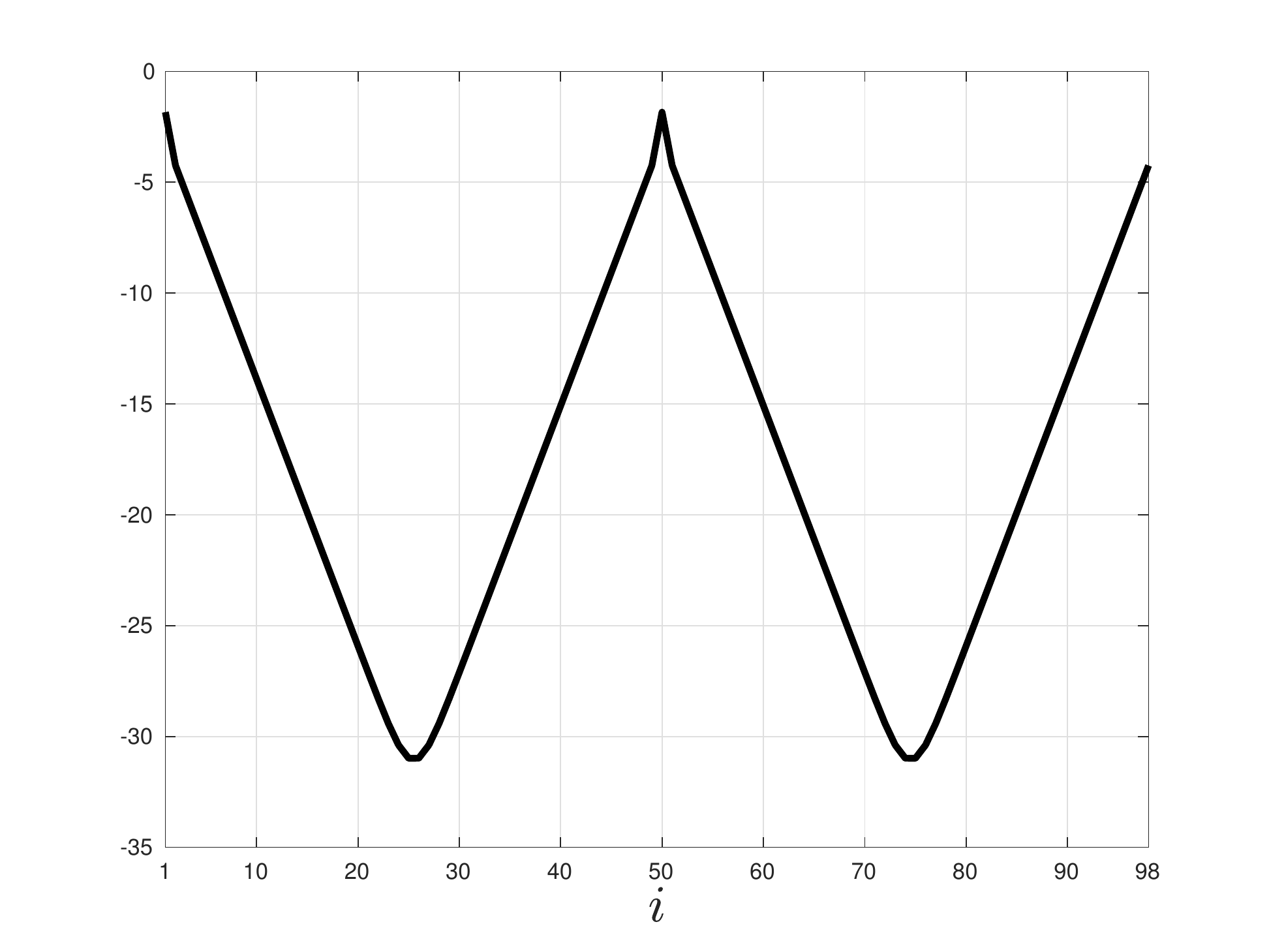}
\caption{$\ln(\phi_i)$ as a function of $i$ for a RFMR with $n=98$ and with rates $\lambda_1=\lambda_{50}=0.3$ and
$\lambda_i=1$, for all other rates . Note that the maximal sensitivities are~$\phi_1,\phi_{50}$, and that the
sensitivities decrease as we move away from sites~$1$ and~$50$ (recall that the topology  is circular).
}\label{fig:sens_n98_2}
\end{center}
\end{figure}

Eq.~\eqref{eq:R_sens} implies that 
\be\label{eq:sens_ratio}
\frac{\phi_i}{\phi_j}
=\frac{\zeta_i\zeta_{i+1}}{\zeta_j\zeta_{j+1}}  \left (\frac{\lambda_j}{\lambda_i}\right)^{3/2}, \quad
i,j\in\{1,\dots,n\}, 
\ee 
that is, the ratio between any two sensitivities is determined by the corresponding Perron
eigenvalue components and the corresponding rates. One may expect that the highest sensitivity will correspond to the
minimal rate, but~\eqref{eq:sens_ratio} shows that this is not necessarily so. The next example demonstrates this.

\begin{Example}
Consider a RFMR with dimension $n=7$ and rates:
\[
\lambda = \begin{bmatrix} 1 & 1.1 & 0.55 & 1.4 & 1.3 & 0.95 & 0.6 \end{bmatrix}'.
\]
In this case, $R^*=0.2213$. Using~\eqref{eq:R_sens} yields the sensitivities:
\[
\phi = \begin{bmatrix} 0.0355 &   0.0288 &   0.0774  &  0.0129  &  0.0124  &  0.0298  &  0.0820 \end{bmatrix}'.
\]
Note that although the minimum rate is~$\lambda_3$, the maximal sensitivity is~$\phi_7$. This implies that
increasing~$\lambda_7$ by some small value~$\varepsilon>0$
 will increase~$R^*$ more than the increase due to increasing any other rate by~$\varepsilon$.
For example, increasing~$\lambda_3$ by~$0.05$ (and leaving all other rates unchanged) yields~$R^*=0.2248$, while
increasing $\lambda_7$ by $0.05$ instead  (and leaving all other rates unchanged) yields~$R^*=0.2251$.
\end{Example}

The spectral approach can also be used to derive theoretical results on the sensitivities. The next three results
demonstrate  this.

\begin{Proposition}\label{prop:si01}
The sensitivities satisfy $0<\phi_i\le 1$ for all~$i=1,\dots,n$.
\end{Proposition}

This implies that
 an increase [decrease] in any of the rates by~$\varepsilon$
  increases [decreases]  the optimal steady-state production rate by no more than~$\varepsilon$.

\begin{Proposition}\label{prop:sens_homog}
Consider a RFMR with dimension~$n$ and homogeneous rates~\eqref{eq:homog_rates}. Then
\[
\phi_i = \frac{1}{4n}, \quad i=1,\dots,n.
\]
\end{Proposition}
This means that in the homogeneous case, all the sensitivities are equal. This is of course expected, as the circular
topology of the sites implies that all the rates have the same effect on~$R^*$. Furthermore, the sensitivities decrease
with~$n$, i.e. in a longer ring each rate has a smaller effect on~$R^*$.

Assume now that the RFMR rates satisfy
\be\label{eq:sym_rates}
\lambda_i=\lambda_{n-i}, \quad i=1,\dots,n-1,
\ee
i.e. the rates are \emph{symmetric}. Note that since all indexes are interpreted modulo $n$, it is enough that~\eqref{eq:sym_rates} holds for some cyclic permutation of the rates. For example, for $n=3$ the rates are symmetric if at least two of the rates $\lambda_1,\lambda_2,\lambda_3$ are equal.

\begin{Proposition}\label{prop:sym_sens}
Consider a RFMR with dimension $n$ and symmetric rates~\eqref{eq:sym_rates}. Then 
\[
\phi_i=\phi_{n-i}, \quad i=1,\dots,n-1.
\]
\end{Proposition}

\begin{Example}
Consider a RFMR with dimension $n=6$ and rates $\lambda_1=\lambda_5=1$, $\lambda_2=\lambda_4=1.2$, $\lambda_3=0.8$ and $\lambda_6=1.5$. Note that these rates satisfy~\eqref{eq:sym_rates}. The sensitivities are:  
\[
\phi = \begin{bmatrix} 0.0408   & 0.0388  &  0.0804  &  0.0388  &  0.0408  &  0.0200 \end{bmatrix}',
\]
and it may be observed that $\phi_i=\phi_{6-i}$, $i=1,\dots,5$.
\end{Example}

\subsection{Optimizing the production rate}
Any set of rates~$\lambda=(\LMDR)$ induces an optimal density~$s^*$ and the~RFMR initialized with this total density
yields a maximal production rate~$R^*$ (with respect to all other initial densities). 
This yields a mapping~$\lambda \to R^*(\lambda)$. 
 Now suppose that we have some set, denoted by~$\Omega$, of~$n$-dimensional vectors with positive entries. 
Every vector from $\Omega$ can be used as a set of rates~$\lambda$ for the~RFMR, and thus yields a value~$R=R^*(\lambda)$. 
A natural question is: determine a vector~$\eta \in \Omega $ that yields the maximal value, that is,
\[
				R(\eta) = \max_{\lambda \in \Omega} R^*(\lambda).
\]

In the context of translation, this means that a circular mRNA with rates~$\eta$, 
  initialized with total density~$s^*(\eta)$, will yield a steady-state production rate that is higher than that obtained for all the other options for the rate vector in~$\Omega$ (regardless
of the initial total density in these other circular mRNAs).

The next result is essential for efficiently analyzing the maximization of $R^*$ with respect to (w.r.t.) its rates.
\begin{Proposition}\label{prop:R_concave}
Consider a RFMR with dimension $n$. The mapping $\lambda:=(\LMDR)\mapsto R^*(\lambda)$ is strictly concave on $\R^n_{++}$.
\end{Proposition}

For example, for~$n=2$ it is straightforward to show that
\[
R^*(\lambda_1,\lambda_2)=\frac{\lambda_1 \lambda_2}{(\sqrt{\lambda_1}+\sqrt{\lambda_2})^2}.
\]
Fig.~\ref{fig:rfmr_n2_Rs} depicts $R^*(\lambda_1,\lambda_2)$ as a function of its parameters. It may be observed that this is a strictly concave function on $\R^2_{++}$.

\begin{figure}[t]
 \begin{center}
  \includegraphics[scale=0.5]{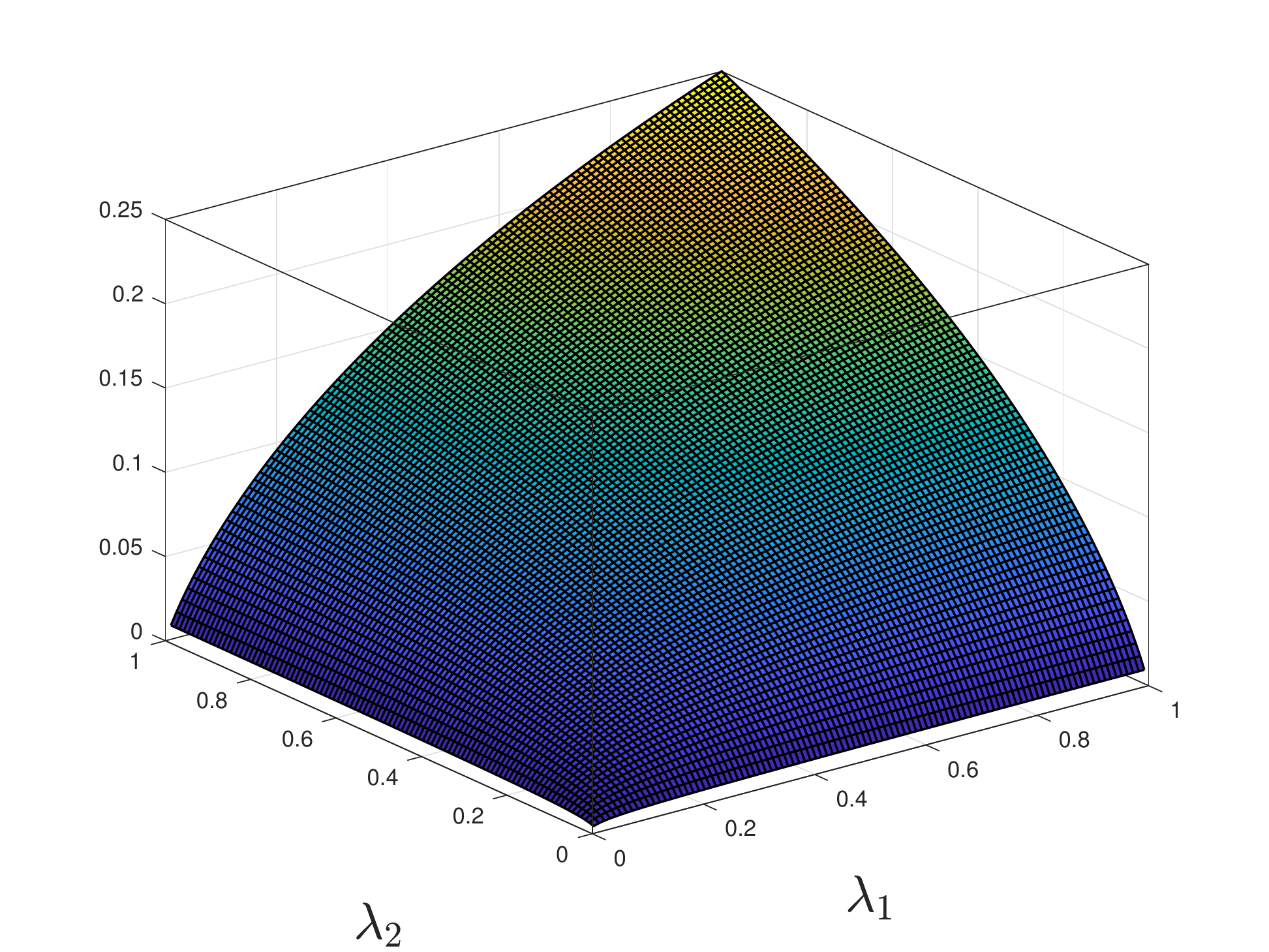}
\caption{$R^*(\lambda_1,\lambda_2)$ in RFMR with $n=2$ as a function of its parameters.}\label{fig:rfmr_n2_Rs}
\end{center}
\end{figure}

The sensitivity analysis of $R^*$, and its strict concavity w.r.t the rates, have important 
 implications to the problem of
 optimizing    the steady-state production rate in the~RFMR w.r.t the rates~$\lambda $. We now explain this using 
a specific   optimization problem. First note that to make the problem meaningful every rate must be bounded 
above. Otherwise, the optimal solution will be to take this  rate to infinity. We thus 
consider the following constrained optimization problem.
\begin{Problem}\label{prob:const_opt}
Consider a RFMR with dimension $n$. Given the parameters $w_1,\dots,w_n, b>0$, maximize $R^*=R^*(\LMDR)$ with respect to the parameters $\LMDR$, subject to the constraints
\begin{align}\label{eq:const_opt}
&\sum_{i=1}^n w_i \lambda_i\le b,  \\
& \LMDR > 0. \nonumber
\end{align}
\end{Problem}
In other words, the problem is to maximize $R^*$ w.r.t. the rates, under the constraints that the rates are positive and their weighted sum is bounded  by~$b$. The weights~$w_i$s can be used to provide different weighting to the different rates, and $b$ represents  the ``total biocellular  budget''.  By Prop.~\ref{prop:si01}, the optimal solution always satisfies the constraint in~\eqref{eq:const_opt} with equality. Note that a similar optimization problem was defined and analyzed in the context of the RFM in~\cite{rfm_max}.

In the context of mRNA translation, each $\lambda_i$ depends on
 the availability of   translation resources that affect codon decoding times, such as tRNA molecules, amino acids, elongation factors, and Aminoacyl tRNA synthetases. These resources are limited 
 as generating them consumes significant amounts of cellular energy. 
They are also correlated. For example, a large $\lambda_i$ may imply large consumption of 
certain~tRNA molecules by site $i$, depleting the availability of tRNA molecules to the other sites. 
Thus, the first (affine) constraint in~\eqref{eq:const_opt} describes the limited and shared translation resources, whereas
 $b$ describes the total available biocellular budget.

By Prop.~\ref{prop:R_concave}, the objective function in Problem~\ref{prob:const_opt} is strictly concave, and since the constraints are affine, Problem~\ref{prob:const_opt} is a \emph{convex optimization problem}~\cite{convex_boyd}. Thus, it admits a \emph{unique} solution.
 We denote the optimal solution of Problem~\ref{prob:const_opt} by  $\lambda^{co}:=(\lambda_1^{co},\dots,\lambda_n^{co})$, and the corresponding maximal (now in the sense of total density and transition rates) steady-state production rate by $R^{co}$ (where~$co$ denotes constrained 
optimization). This means that for a RFMR with dimension~$n$, $R^{co}$ is the maximal steady-state production rate over all the    rates satisfying the constraints~\eqref{eq:const_opt} and all possible total initial  densities.

The convexity also implies that the solution 
 can be determined efficiently using numerical algorithms that scale well with~$n$.
To demonstrate this, we wrote a simple and \emph{unoptimized} MATLAB program  
(that is guaranteed to converge because of the convexity) for solving this optimization problem and ran it on 
a   MAC laptop with a~$2.6$ GHz Intel core~i7 processor.
 As an example, for~$n=100$ and the (arbitrarily chosen) weights~$w_i = 1 + 0.4 \sin(2 \pi i / 100 )$, $i=1,\dots,100$, and~$b=1$, 
the optimal solution was found   
 after~$11.7$ seconds.

The affine constraint in~\eqref{eq:const_opt} includes a possibly different weight for each of the rates. For example, if~$w_2$ is much larger than the other weights then this means that any small increase in~$\lambda_2$ will greatly increase the total weighted sum, thus typically forcing the optimal value~$\lambda_2^{co}$ to be small.
In the special case where all the~$w_i$s are equal the formulation
 gives equal preference to all the rates, so if the corresponding optimal solution satisfies $\lambda_j^{co}>\lambda_i^{co}$, for some~$i,j$, then this implies that, in the context of maximizing~$R^*$, $\lambda_j$ is ``more important''
 than~$\lambda_i$. We refer to this case as the \emph{homogeneous constraint} case and assume, without loss of generality, that~$w_i=1$ for all~$i$. 
 Note that by~\eqref{eq:hom_deg1} we can always assume, without loss of generality, that $b=1$.

\begin{Proposition}\label{prop:const_opt_hom}
Consider Problem~\ref{prob:const_opt} with $w_1=\cdots=w_n=b=1$, i.e. the affine constraint is  
\be\label{eq:const_opt_hom}
\sum_{i=1}^n \lambda_i=1.
\ee
Then the optimal solution is $\lambda_i^{co} =1/n$ for all~$i$. The~RFMR with these rates satisfies~$s^*=n/2$, 
$ e_i^{co}=1/2 $ for all~$i$,  
and~$R^{co} =1/(4n)$. 
\end{Proposition}

\begin{Remark}
In view of the Kuhn--Tucker theorem~\cite{convex_boyd}, the necessary and sufficient condition for
optimality of
$\lambda$ in Problem~\ref{prob:const_opt} with homogeneous weights is that the sensitivity $\phi_i=\frac{\partial R^*}{\partial \lambda_i}(\lambda^*)$
does not depend on the index $i$.
\end{Remark}

\section{Discussion}\label{sec:disscus}

We considered a deterministic model for translation along a circular~mRNA.
The behavior of this model depends on the transition rates between the sites 
and on the value~$s:=\sum_{i=1}^n x_i(0)$, that is, the initial total density along the 
ring. The total density is conserved, so~$\sum_{i=1}^n x_i(t)=s$ for all~$t\geq 0$.

We  derived a spectral representation for the
steady-state density and production rate for the case  where the initial density is~$s^*$, i.e. the density yielding a maximal 
steady-state production rate.  
  In fact, the proof of  Thm.~\ref{thm:rfmr_lin} (see the Appendix)  shows that we can
interpret the optimal density~RFMR as a dynamical system that ``finds'' the Perron eigenvalue and eigenvector of a
certain periodic Jacobi matrix.

The spectral representation for the~RFMR provides
a powerful    framework for analyzing the RFMR when initialized with the optimal total density $s^*$. In addition to providing an efficient and numerically stable
manner for computing
 the optimal steady-state production rate and steady-state density, it allows to  efficiently compute
 the sensitivity of the optimal steady-state production rate to perturbations in the rates. 
This is important as conditions in the cell are inherently
 stochastic, and thus sensitivity analysis must accompany the steady-state
   description. 
	
	Furthermore, using the spectral representation, it was shown that the steady-state production rate with optimal density is a strictly concave function of the RFMR rates. The translation machinery in the cell is affected by different kinds of mutations (e.g. synonymous codon mutations, duplication of a tRNA gene, etc.). The strict concavity result thus suggest that the selection of mutations that increase fitness indeed converges towards the unique optimal parameter values (by a simple ``hill-climbing'' evolution process).  The strict concavity implies that
	 given an affine (and more generally convex) constraint on the rates, that represents  limited and shared translation resources, the unique optimal set of rates can be determined efficiently even for (circular) mRNAs with a  large number of codons.

Obtaining an optimal production rate is an important problem in synthetic biology and biotechnology. Examples include optimal synonymous codon mutations of an endogenous gene, and optimal translation efficiency and protein levels of heterologous genes in a new host~\cite{Plotkin2011,Tuller2010,Gustafsson2004,Kimchi-Sarfaty2013}. These genes compete with endogenous genes for the available  translation resources, as consuming too much resources by the heterologous gene may kill the host~\cite{Plotkin2011,Tuller2010}. Thus, any realistic optimization of the protein production rate should not consume too many resources, as otherwise the fitness of the host may be significantly reduced. These considerations seems to fit well with the affine-constrained optimization problem presented and analyzed here. 

We also showed that the spectral representation of  the RFM follows as a special case of the representation given here
for the~RFMR.   However, it seems that a better understanding of the link between the RFM and the RFMR requires further study.  
Our results suggest several other interesting directions for future research. One such direction is finding  special cases,
besides the one described in Example~\ref{exa:homg}, where the Perron eigenvalue and eigenvector
of~$A(\lambda_1,\dots,\lambda_n)$ are explicitly known. Another possible direction is the analysis of the dual of the optimization problem defined by Problem~\ref{prob:const_opt}. Specifically, does the dual problem has any interesting biological interpretation in the context of mRNA translation, and does its analysis provides more insight into optimizing translation?

Finally,  TASEP with periodic boundary conditions has been  used to model many transport phenomena
including   traffic flow and pedestrian dynamics~\cite{TASEP_book, waldau2007pedestrian}. We believe that the spectral representation of the~RFMR with optimal density may be useful also for analyzing these transport applications. 

\section*{Acknowledgments}
The research of~YZ is partially supported by the Edmond J. Safra Center for Bioinformatics at Tel Aviv University.
 The research of~AO is partially supported by the Russian Foundation for Basic Research, grant 17-08-00742. 
The research of~MM is partially supported by research grants from the Israeli Ministry of Science, Technology \& Space, the US-Israel  Binational Science Foundation, and  the  Israeli Science Foundation.

\section*{Author Contributions Statement}
YZ, AO, and MM performed the research  and wrote the paper.

\section*{Data Availability Statement} All the relevant data is included in the manuscript. 

\section*{Competing Financial Interests Statement} The authors declare no competing financial interests.

\section*{Appendix - Proofs}\label{sec:proofs}

\begin{proof} [\bf Proof of Thm.~\ref{thm:rfmr_lin}]
Pick~$n>2$ and parameters~$c_1,\dots,c_{n-1}> 0$, and~$c_n \geq 0$.
Consider  the~$n\times n$ periodic Jacobi matrix:
\[
J:=\begin{bmatrix}
0& c_1 & 0 &0& \dots& 0 & 0&0  &c_n \\
c_1& 0 & c_2 &0& \dots& 0 & 0&0&0 \\
0& c_2 &0&c_3 & \dots& 0 & 0 &0&0\\
&&\vdots\\
0 & 0  &0& 0  & \dots& 0 & c_{n-2} &0&c_{n-1}\\
c_n & 0  &0& 0  & \dots& 0&0&c_{n-1}&0
\end{bmatrix}.
\]
  Note that~$J$ is irreducible and (componentwise) non-negative.
Let~$\sigma>0$ denote that Perron eigenvalue of~$J$ and let~$\zeta\in\R^n_{++}$ denote the corresponding eigenvector.
 The equation~$J\zeta=\sigma \zeta$ yields
\begin{align}\label{eq:eig_sys}
c_1 \zeta_2+c_n \zeta_n&=\sigma \zeta_1,  \nonumber\\
c_1 \zeta_1+c_2 \zeta_3&=\sigma \zeta_2,  \nonumber \\
c_2 \zeta_2+c_3 \zeta_4&=\sigma \zeta_3,  \nonumber\\
&\vdots \nonumber \\
c_{n-2} \zeta_{n-2}+c_{n-1}\zeta_n&=\sigma \zeta_{n-1}, \nonumber \\
c_n\zeta_1+c_{n-1} \zeta_{n-1}&=\sigma \zeta_n .
\end{align}
  Define
\be\label{eq:assump11} d_i:= \frac{c_i \zeta_{i+1}}{\sigma \zeta_i},\quad i=1,\dots,n. \ee Note that since the indexes
are interpreted modulo $n$, Eq.~\eqref{eq:assump11} implies in particular that \be\label{eq:assump1} d_n=   \frac{c_n
\zeta_{1}}{\sigma \zeta_n}. \ee

Then~\eqref{eq:eig_sys} yields:
\begin{align}\label{eq:port}
 \sigma^{-2}&=c_n^{-2} d_n(1-d_1),  \nonumber \\
 \sigma^{-2}&=c_1^{-2} d_1(1-d_2), \nonumber \\
 \sigma^{-2}&=c_2^{-2} d_2(1-d_3),  \nonumber \\
&\vdots   \\
\sigma^{-2}&=c_{n-2}^{-2} d_{n-2}(1-d_{n-1}),  \nonumber \\
 \sigma^{-2}&=c_{n-1}^{-2} d_{n-1}(1-d_n).\nonumber
\end{align}
Also, it follows from~\eqref{eq:assump11} that~$\prod_{i=1}^n d_i= \sigma^{-n}\prod_{i=1}^n c_i$, and
from~\eqref{eq:port} that~$\prod_{i=1}^n (1-d_i)= \sigma^{-2n}\frac{\prod_{i=1}^n c_i^2}{\prod_{i=1}^n d_i}$, and
combining these two equations yields \be\label{eq:podf}
 \prod_{i=1}^n d_i = \prod_{i=1}^n (1-d_i) .
\ee Note that all the derivations above hold for any real eigenvalue of~$J$ and its corresponding eigenvector (assuming
all its entries are non zero so that~\eqref{eq:assump11} is well-defined), but since the Perron eigenvector is the only
eigenvector in the first orthant~\cite{matrx_ana}, all the~$d_i$s are positive only for the Perron eigenvalue and eigenvector.

Now consider  a RFMR with dimension $n$ and rates $\lambda_i:=c_i^{-2}$, $i=1,\dots,n$, that is:
\begin{align}\label{eq:cnrfmr}
                        \dot x_1&=c_n ^{-2} x_n (1-x_1)-c_1^{-2} x_1(1-x_2) \nonumber \\
                        \dot x_2&=  c_1^{-2} x_1(1-x_2) -c_2^{-2} x_2(1-x_3) \nonumber  \\
                        &\vdots\\
              \dot x_{n-1}&=c_{n-2}^{-2} x_{n-2}(1-x_{n-1}) -c_{n-1}^{-2} x_{n-1}(1-x_n)  \nonumber \\
              \dot x_{n}&=c_{n-1}^{-2} x_{n-1}(1-x_{n}) -c_{n}^{-2} x_{n}(1-x_1)  \nonumber .
\end{align}
We already know that this system converges to a steady-state~$e\in C^n$, that is,
\begin{align*}
R=c_n ^{-2} e_n (1-e_1)=c_1^{-2} e_1(1-e_2) =\dots=c_{n-1}^{-2} e_{n-1}(1-e_{n}).
\end{align*}
Comparing this with~\eqref{eq:port} shows that~$e_i=d_i$ for all~$i$, and that  the steady-state production rate
is~$R=\sigma^{-2}$. Furthermore,~\eqref{eq:podf} implies that~$\prod_{i=1}^n e_i=\prod_{i=1}^n (1-e_i)$, so we conclude
that the steady-state satisfies condition~\eqref{eq:swee(1-e)} that describes the unique  optimal steady-state (i.e. the steady-state production rate that corresponds to the unique optimal total density $s^*$). This
proves the first two equations in~\eqref{eq:spect_rep}. Finally, since the total density is conserved, it is equal
to~$\sum_{i=1}^n e_i$. 
This completes the proof of Thm.~\ref{thm:rfmr_lin}.
\end{proof}

\begin{proof} [\bf Proof of Proposition~\ref{prop:sense}]
By Thm.~\ref{thm:rfmr_lin}, \be\label{eq:potr} \phi_i=\frac{\partial  }{\partial  \lambda_i}\sigma^{-2} =-2
\sigma^{-3}\frac{\partial \sigma }{\partial  \lambda_i} . \ee By known results from linear algebra (see,
e.g.,~\cite{magnus85}), the sensitivity of the Perron root of~$A$ with respect to a change in~$\lambda_i$ is
\[
\frac{\partial}{\partial \lambda_i}\sigma=\frac{\zeta' \left( \frac{d}{d \lambda_i}A \right ) \zeta}{\zeta' \zeta} .
\]
Only the entries~$a_{i,i+1}=a_{i+1,i}= \lambda_i^{-1/2}$
  depend on~$\lambda_i$, so
\[
\frac{\partial}{\partial \lambda_i}\sigma=\frac{ -\zeta_i\zeta_{i+1} \lambda_i^{-3/2}}{ \zeta' \zeta},
\]
 and combining this with~\eqref{eq:potr}
 proves~\eqref{eq:R_sens}.
\end{proof}


\begin{proof}[\bf Proof of Prop.~\ref{prop:si01}]
Since $\sigma>0$ and $\zeta\in\R^n_{++}$,   $\phi_i>0$ for all~$i$. To prove the upper bound, perturb~$\lambda_i$
to~$\bar \lambda_i:=\lambda_i+\varepsilon$, with~$\varepsilon>0$ sufficiently small. This yields a perturbed matrix $\bar
A$ that is identical to $A$ except for entries $(i,i+1)$ and $(i+1,i)$ that are
\[
\bar \lambda_i^{-1/2}=(\lambda_i + \varepsilon)^{-1/2}=\lambda_i^{-1/2} -\frac{\varepsilon \lambda_i^{-3/2}}{2} +
o(\varepsilon),
\]
where~$o(\varepsilon)$ denotes a function~$f(\varepsilon)$ satisfying~$\lim_{\varepsilon\to 0}
\frac{f(\varepsilon)}{\varepsilon}=0$. This means that $\bar A  = A + P$, where $P\in\R^{n\times n}$ is a matrix with
zero entries except for entries $(i,i+1)$ and $(i+1,i)$ that are equal to~$-\frac{\varepsilon \lambda_i^{-3/2}}{2} +
o(\varepsilon)$. By Weyl's inequality~\cite{matrx_ana}, $\rho(\bar A) \ge \rho(A) - \frac{\varepsilon \lambda_i^{-3/2}}{2
}+o(\varepsilon)$, where $\rho(Q)$ denotes the maximal eigenvalue of a  symmetric matrix $Q$. This means
that~$\frac{\partial \rho(A)}{\partial \lambda_i} \ge  - \frac{\lambda_i^{-3/2}}{2}+\frac{o(\varepsilon)}{\varepsilon}$,
thus $\phi_i \le (R^*/\lambda_i)^{3/2}$. Since $R^* \le \lambda_i$, it follows that $\phi_i \le 1$ for all~$i$.
\end{proof}

\begin{proof}[\bf Proof of Prop.~\ref{prop:sens_homog}]
Consider a RFMR with homogeneous rates~\eqref{eq:homog_rates}. Then by Example~\ref{exa:homg}, $\zeta_i=1$, $i=1,\dots,n$, and $\sigma=2\lambda_c^{-1/2}$, and plugging these in~\eqref{eq:R_sens} completes the proof.

\end{proof}

\begin{proof}[\bf Proof of Prop.~\ref{prop:sym_sens}]
We require the following result.

\begin{Proposition}\label{prop:sym_zeta}
Consider the RFMR with dimension~$n$ and symmetric rates. Then~$\zeta_i=\zeta_{n+1-i}$, $i=1,\dots,n$.
\end{Proposition}

\begin{proof}[\bf Proof of Prop.~\ref{prop:sym_zeta}]
Consider first the case~$n$ even. Let $Q\in\R^{(n/2)\times (n/2)}$ be a reversal matrix, i.e. a matrix of zeros 
except for the counter-diagonal (i.e.   entries $(i,\frac{n}{2}-i+1)$, $i=1,\dots,n/2$) that is all ones. For example, for $n=4$, 
\[
Q=\begin{bmatrix} 0 & 1 \\ 1 & 0 \end{bmatrix}.
\]
Note that given any arbitrary vector $v=\begin{bmatrix} v_1 & v_2 &\cdots &v_{n/2} \end{bmatrix}'\in\R^{n/2}$, $Qv=\begin{bmatrix} v_{n/2} & v_{(n/2)-1} & \cdots & v_1 \end{bmatrix}'$.

Since the rates satisfy~\eqref{eq:sym_rates}, the~$n\times n $ matrix $A$ has the form
\[
A=\begin{bmatrix} A_1 & A_2 \\ Q A_2 Q & Q A_1 Q \end{bmatrix},
\]
where $A_1\in\R^{(n/2)\times (n/2)}_+$ is a matrix of zeros except for
 the super-diagonal and the sub-diagonal, which are both equal to $(\lambda_1^{-1/2},\dots,\lambda_{(n/2)-1}^{-1/2})$, and $A_2\in\R^{(n/2)\times(n/2)}_+$ is a matrix of zeros except for
   entry $(1,n/2)$ that is     $\lambda_n^{-1/2}$, and   entry $(n/2,1)$ that is     $\lambda_{n/2}^{-1/2}$. Decompose the Perron eigenvector~$\zeta$ of~$A$
as~$\zeta^1:=\begin{bmatrix} \zeta_1 & \dots & \zeta_{n/2} \end{bmatrix}'$ and $\zeta^2:=\begin{bmatrix} \zeta_{(n/2)+1} & \dots & \zeta_{n} \end{bmatrix}'$. 

Let~$\rho(W)$ denote the spectral radius of a matrix~$W$. 
Since~$A_1$ is a principal submatrix
 of the (componentwise) nonnegative matrix~$A$, $\rho(A_1)\leq \rho(A)$ (see~\cite[Ch.~8]{matrx_ana}). 
Assume for the moment that~$\rho(A_1) = \rho(A)$. 
Then using the fact that~$QQ=I$, that is~$Q=Q^{-1}$, we conclude
 that~$\rho\left(  \begin{bmatrix} A_1 & 0 \\ 0 & Q A_1 Q \end{bmatrix} \right)=\rho(A)$.
This means that the matrices
$\begin{bmatrix} A_1 & 0 \\ 0 & Q A_1 Q \end{bmatrix}$ and~$ \begin{bmatrix} A_1 & A_2 \\ Q A_2 Q & Q A_1 Q \end{bmatrix}$
have the same Perron root, but this     contradicts Prop.~\ref{prop:si01}. We conclude that
\be\label{eq:smlll}
							\rho(A_1)<\rho(A)=\sigma. 
\ee

 The equation $A \zeta=\sigma \zeta$ yields
\begin{align*}
A_1 \zeta^1 + A_2 \zeta^2 &= \sigma \zeta^1, \\
Q A_2 Q \zeta^1 + Q A_1 Q \zeta^2 &= \sigma \zeta^2.
\end{align*}
Multiplying both sides of the second equation by $Q$, noting that $QQ=I$,
 and rearranging yield
\begin{align}\label{eq:z1z2}
A_1 \zeta^1 + A_2 \zeta^2 &= \sigma \zeta^1, \nonumber \\
A_1 Q \zeta^2 + A_2 Q \zeta^1 &= \sigma Q \zeta^2.
\end{align}
Subtracting  the second equation from the first and using again the fact that~$QQ=I$ yields
\be\label{eq:poyte}
							(A_1-A_2Q-\sigma I)(\zeta^1-Q\zeta^2)=0.
\ee
Combining this with~\eqref{eq:smlll} and the fact that~$A_2Q$ is (componentwise) nonnegative implies that~$\zeta^1 = Q\zeta^2$,  i.e. $\zeta_i=\zeta_{n+1-i}$, $i=1,\dots,n$. This completes the proof for the case~$n$ even. The proof when~$n$ is odd is very similar and therefore omitted. 
\end{proof}

Now the proof of Prop.~\ref{prop:sym_sens}   follows from combining~\eqref{eq:R_sens}, Thm.~\ref{thm:rfmr_lin}, and~Prop.~\ref{prop:sym_zeta}.
\end{proof}

\begin{proof}[\bf Proof of Prop.~\ref{prop:R_concave}]
Indeed, the  map $\lambda=(\lambda_1,\dots,\lambda_{n})\mapsto
 A(\lambda)$, where $A(\lambda)$ is given in~\eqref{eq:matA}, from $\R^n_{++}$ is convex, meaning that
the matrix inequality
\begin{equation}\label{conv}
    \frac12\left(A(\lambda')+A(\lambda'')\right)\geq A(\frac12(\lambda'+\lambda''))
\end{equation} holds elementwise for any arbitrary $\lambda',\lambda''\in\R^n_{++}$. This immediately follows from the convexity of the real function
$\lambda\mapsto\lambda^{-1/2}$. The Perron--Frobenius theorem implies the corresponding inequality for the
Perron eigenvalue~\cite{matrx_ana}
\begin{equation}\label{conv_rho}
    \frac12\left(\sigma(A(\lambda'))+\sigma(A(\lambda''))\right)\geq \sigma(A(\frac12(\lambda'+\lambda''))),
\end{equation}
where the inequality~\eqref{conv_rho} is strict if $\lambda'\neq\lambda''$. Thus $\sigma(\lambda)$ is a strictly convex function. In view of the basic identity
$R^*=\sigma^{-2}$ in~\eqref{eq:spect_rep}, it follows that $R^*(\lambda)$ is a strictly concave function.
\end{proof}

\begin{proof}[\bf Proof of Prop.~\ref{prop:const_opt_hom}]
We know that Problem~\ref{prob:const_opt} admits a unique optimal solution~$\tilde \lambda$. Consider the cyclic shift
$\bar \lambda_{i}=\tilde \lambda_{i+1}$, $i=1,\dots,n$, where the indices are taken modulo $n$. Note that $\sum_{i=1}^n \bar \lambda_i=\sum_{i=1}^n \tilde \lambda_i =1$, so~$\bar \lambda$ also satisfies the constraint~\eqref{eq:const_opt_hom}. The
matrices $A(\bar \lambda)$ and $A(\tilde \lambda)$ have the same spectrum.  Since the optimal solution is unique, $\overline\lambda=\tilde \lambda$.
We conclude that the optimal transition rates~$\tilde \lambda_i$  are all equal, and thus~$\lambda_i^{co}:=1/n$,  $i=1,\dots,n$. 
By Example~\ref{exa:homg}, $R^{co}=1/(4n)$, and $e^{co}_i=1/2$, $i=1,\dots,n$.
\end{proof}

\end{document}